\documentclass[12pt]{amsart}

\usepackage{latexsym,amssymb,amsfonts,amsmath,mathrsfs,hyperref,mathtools}
\usepackage{pifont}%
\usepackage{enumerate}
\hypersetup{colorlinks,citecolor=NavyBlue,filecolor=black,linkcolor=black,urlcolor=black}
\usepackage[dvipsnames]{xcolor}
\usepackage{comment} 
\usepackage[export]{adjustbox}
\usepackage{graphicx}

\newtheorem{que}{Question}

\newcommand{\parset}{
	\setlength{\parskip}{3mm}
  	\setlength{\parindent}{0mm}}
	
\parset
	
  \newcommand{\Exp}{{\mathbb{E}}}

  \DeclareMathOperator{\ccb}{iccb}  
  \DeclareMathOperator{\jcb}{jcb}
  
  \newcommand{\norm}[1]{\|#1\|}

  \newcommand{\R}{\mathbb{R}} 
  \newcommand{\C}{\mathbb{C}} 
  \newcommand{\N}{\mathbb{N}} 
  \newcommand{\Z}{\mathbb{Z}} 
  \newcommand{\pmset}[1]{\{-1,1\}^{#1}} 

  \newcommand{\st}{\,\mid\,} 

  \newcommand{\eps}{\varepsilon}

  \DeclareMathOperator{\sdp}{SDP}
  
  \newcommand{\id}{I}

  \newcommand{\beq}{\begin{equation}}
  \newcommand{\eeq}{\end{equation}}
  \newcommand{\beqn}{\begin{equation*}}
  \newcommand{\eeqn}{\end{equation*}}
  \newcommand{\beqr}{\begin{eqnarray}}
  \newcommand{\eeqr}{\end{eqnarray}}
  \newcommand{\beqrn}{\begin{eqnarray*}}
  \newcommand{\eeqrn}{\end{eqnarray*}}
  \newcommand{\bmline}{\begin{multline}}
  \newcommand{\emline}{\end{multline}}
  \newcommand{\bmlinen}{\begin{multline*}}
  \newcommand{\emlinen}{\end{multline*}}

  \theoremstyle{plain}
  \newtheorem{theorem}{Theorem}[section]
  \newtheorem{lemma}[theorem]{Lemma}
  \newtheorem{proposition}[theorem]{Proposition}
  
  \newtheorem{corollary}[theorem]{Corollary}
  
  \theoremstyle{definition}

  \theoremstyle{remark}
  \newtheorem{remark}[theorem]{Remark}
  
  \renewenvironment{proof}[1][]{
    	\begin{trivlist}
     	\item[\hspace{\labelsep}{\em\noindent Proof#1:\/}]}
     	{{\hfill$\Box$}
    	\end{trivlist}
  }
  
  \makeatletter
  \newtheorem*{rep@theorem}{\rep@title}
  \newcommand{\newreptheorem}[2]{%
  \newenvironment{rep#1}[1]{%
  \def\rep@title{#2 \ref{##1}}%
  \begin{rep@theorem}}%
  {\end{rep@theorem}}}
  \makeatother

  \newreptheorem{lemma}{Lemma}
  \newreptheorem{theorem}{Theorem}
  \newreptheorem{corollary}{Corollary}
  \newreptheorem{proposition}{Proposition}
  \newreptheorem{conjecture}{Conjecture}

\newif\ifnotes\notesfalse

\ifnotes
\usepackage{color}
\definecolor{mygrey}{gray}{0.50}
\newcommand{\notename}[2]{{\textcolor{mygrey}{\footnotesize{\bf (#1:} {#2}{\bf ) }}}}
\newcommand{\noteswarning}{{\begin{center} {\Large WARNING: NOTES ON}\end{center}}}

\else

\newcommand{\notename}[2]{{}}
\newcommand{\noteswarning}{{}}

\fi

\newcommand{\jnote}[1]{{\notename{Jop}{#1}}}

\begin{document}

\title[On  converses to the polynomial method]{On  converses to the polynomial method}

\author{Jop Bri\"{e}t}
\address{CWI \& QuSoft, Science Park 123, 1098 XG Amsterdam, The Netherlands}
\email{j.briet@cwi.nl}

\author{Francisco Escudero Guti\'errez}
\address{CWI \& QuSoft, Science Park 123, 1098 XG Amsterdam, The Netherlands}
\email{feg@cwi.nl}

\thanks{\includegraphics[height=4ex]{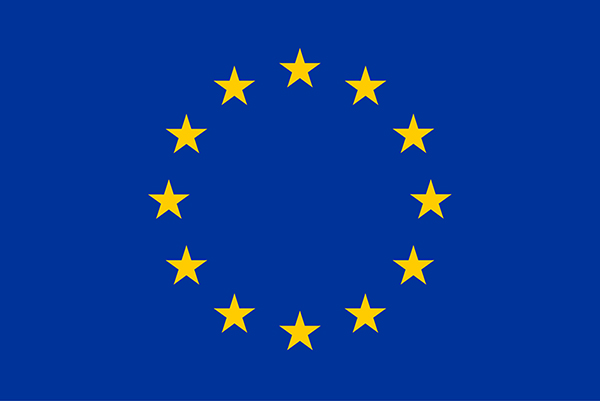} This research was supported by the European Union’s  Horizon 2020 research and innovation programme under the Marie Sk{\l}odowska-Curie grant agreement no. 945045, and by the NWO Gravitation project NETWORKS under grant no. 024.002.003}

\maketitle

\noteswarning
\begin{abstract}
A surprising `converse to the polynomial method' of Aaronson et al.\ (CCC'16) shows that any bounded quadratic polynomial can be computed exactly in expectation by a 1-query algorithm up to a universal multiplicative factor related to the famous Grothendieck constant. A natural question posed there asks if  bounded quartic polynomials can be approximated by $2$-query quantum algorithms. 
Arunachalam, Palazuelos and the first author showed that there is no direct analogue of the result of Aaronson et al.\ in this case. 
We improve on this result in the following ways: First, we point out and fix a small error in the construction that has to do with a translation from cubic to quartic polynomials. Second, we give a completely explicit example based on techniques from additive combinatorics. Third, we show that the result still holds when we allow for a small additive error. For this, we apply an SDP characterization of Gribling and Laurent (QIP'19) for the completely-bounded approximate degree. 
\end{abstract}

\section{Introduction}

A celebrated result of Beals et al.~\cite{polynomialmethod}, known as the \emph{polynomial method} in quantum complexity theory, leverages the problem of lower bounding the quantum query complexity of a Boolean function to lower bounding the approximate degree.
The method is based on the fact that for every  $t$-query quantum algorithm $\mathcal{A}$ that takes an $n$-bit input and returns a sign, there is a real $n$-variable polynomial~$f$ of degree at most~$2t$ such that $f(x)=\mathbb{E}[\mathcal{A}(x)]$ for every~$x$. Here, the expectation is taken with respect to the randomness in the measurement done by~$\mathcal A$.\footnote{We identify a quantum query algorithm with the (random) function representing its output on a given input string.}
In addition to many new lower bounds, this result led to a line of research on possible \emph{converses}, whereby a bounded polynomial~$f$ can be turned into a quantum query algorithm that approximates~$f$ and whose query complexity depends in some reasonably way on the degree of~$f$.
Here,~$f$ is \emph{bounded} if it maps the Boolean hypercube to the interval~$[-1,1]$ and a quantum query algorithm~$\mathcal A$ \emph{approximates}~$f$ if for some constant error parameter $\eps < 1$, we have that $|f(x) - \Exp[\mathcal A(x)]| \leq \eps$ for every~$x$.
For bounded polynomials of degree at most~2, the following  converse was proved in \cite{Aaronson2015PolynomialsQQ}, using a surprising application of the Grothendieck inequality from Banach space theory (we refer to~\cite{pisier2012grothendieck} for an extensive survey on Grothendieck-type inequalities).

\begin{theorem}[Aaronson et al.]\label{query1}
There exists an absolute constant $C>0$ such that the following
holds. For every bounded polynomial~$f$ of degree at most~$2$, there exists a one-query quantum algorithm~$\mathcal A$ such that $\Exp[\mathcal A(x)] = Cf(x)$ holds for every $x\in\{-1,1\}^n$.
\end{theorem}

This ``multiplicative converse'' implies an approximation with  \emph{additive} error at most $1-C$. 
A natural question is if this result generalizes to quartic polynomials and two-query quantum algorithms~\cite[Section~5, Question~1]{Aaronson2015PolynomialsQQ}. Based on the probabilistic method and a new characterization of quantum query algorithms in terms of completely bounded polynomials, a counterexample to a direct analog of Theorem~\ref{query1} was given for quartic polynomials in~\cite{QQA=CBF}.

\begin{theorem}[Arunachalam--Bri\"{e}t--Palazuelos]\label{thm:counter4}
	For any $C>0$, there exist an $n\in \N$ and a bounded quartic $n$-variable polynomial~$f$ such that no two-query quantum algorithm~$\mathcal A$  satisfies $\Exp[\mathcal{A}(x)] = C f(x)$ for every~$x\in\{-1,1\}^n$.
\end{theorem}

However, this result does not exclude the possibility that all bounded quartic polynomials can be (additively) approximated by  two-query quantum algorithms.
Moreover, the result is not constructive, relying on results from random matrix theory to show the existence of such polynomials.
Finally, the result was obtained by transforming a certain random \emph{cubic} polynomial into a quartic polynomial with similar properties.
As we will explain here, the argument given in~\cite{QQA=CBF} to show that there is such a transformation contains an error.
Here, we address these issues as follows:

First, we correct the error in~\cite{QQA=CBF}, showing that Theorem~\ref{thm:counter4} holds as stated.
	
Second, we give a completely explicit example for Theorem~\ref{thm:counter4} using ideas from the field of additive combinatorics that were applied to construct counterexamples to certain far-reaching generalizations of the Grothendieck inequality~\cite{briet2018failure}.
	
Third, we strengthen Theorem~\ref{thm:counter4} by showing that it still holds with a small additive error:
	
	\begin{theorem}\label{thm:NoMult+AddConverse}
		For any $C>0$, there exist an $n\in \N$, an~$\eps >0$ and a bounded quartic $n$-variable polynomial~$f$ such that  no two-query algorithm $\mathcal{A}$ satisfies $|\Exp[\mathcal{A}(x)]-Cf(x)|< \eps$ for every $x\in\{-1,1\}^n$.
	\end{theorem}
	
	This result is an application of a semidefinite-program (SDP) of Gribling and Laurent~\cite{gribling2019semidefinite} for quantum query complexity.
It can be interpreted as an analogue of results on approximate degree based on its linear-programming-based characterization (see for instance~\cite{BunKhotariThaler:2018}).
	To the best of our knowledge, this is the first application of~\cite{gribling2019semidefinite} to prove lower bounds on quantum query complexity.
	As such, we believe it can serve as a first step towards using this SDP to approach other  problems such as proving large separations between approximate degree and quantum query complexity, for example~\cite{aaronson2021open}. 
	The SDP is based on a characterization of quantum query algorithms from~\cite{QQA=CBF} that thus far has proved difficult to use. Surprisingly, in parallel to our work, that was used by Bansal, Sinha and de Wolf~\cite{bansal2022influence} to make progress on a famous problem on the need for structure in quantum speed-ups.
	Their work settled a special case of the Aaronson-Ambainis conjecture~\cite{Aaronsons:2014} showing that polynomials resulting from particular types of quantum algorithms have influential variables.

	In similar vein, we use a basic lower bound on the (real) Grothendieck constant, denoted~$K_G$, based on the CHSH Bell inequality to give an impossibility result for one-query quantum algorithms. That is, we show that there exists a bounded quadratic polynomial~$f$ such that no one-query quantum algorithm approximates~$f$ with error less than $1 - 1/\sqrt{2}$.
	Motivated by this, we pose as an open question whether this can be improved to $1 - 1/K_G$.
	Since the result of~\cite{Aaronson2015PolynomialsQQ} achieves this for bounded bilinear forms, this would give yet another characterization of the Grothendieck constant. 
	Tsirelson's characterization in the context of Bell inequalities~\cite{Tsirelson} being a famous example in quantum information theory, for instance.

\section{Preliminaries}\label{sec:preliminaries}

Unless stated otherwise, below~$C$ will stand for an absolute positive  constant whose value may change from line to line.
All polynomials are assumed to be real and multivariate.
A homogeneous polynomial is referred to as a \emph{form}.
A polynomial is multilinear if each variable appears with degree at most~1. 
Given an $n$-variate polynomial~$f$ and $p\in [1,\infty)$, define
\begin{align*}
\|f\|_p &= \big(\Exp_{x\in \pmset{n}}f(x)^p\big)^{\frac{1}{p}}\\
\|f\|_\infty &= \max_{x\in \pmset{n}}|f(x)|.
\end{align*}
We also define the following ``commutative version'' of a completely bounded norm:
\beqn
\|f\|_{\ccb} = \sup_{d\in \N}\Big\{\big\|f(A_1,\dots,A_n)\| \st A_i\in \C^{d\times d},\, \|A_i\|\leq 1,\, [A_i,A_j] = 0\Big\},
\eeqn
where the norms on the right-hand side are the usual operator norms.\footnote{The notation iccb stands for ``identical commutative completely bounded'', where the word identical distinguishes it from another natural variant of the completely bounded norm of a polynomial.}

The following lemma~\cite[Theorem~1.3, Proposition~4.4]{QQA=CBF} relates quantum query algorithms to completely bounded polynomials.

\begin{lemma}\label{lem:qalgccb}
Let~$\mathcal A$ be a $t$-query quantum algorithm.
Then, there exists an $(n+1)$-variate form~$f$ of degree~$2t$ such that $\|f\|_{\ccb} \leq 1$ and which satisfies $f(x,1) = \Exp[\mathcal A(x)]$ for every $x\in \pmset{n}$.
\end{lemma}

We will also use a quantity associated specifically with multilinear cubic forms, that is polynomials of the form:
\beq\label{eq:3form}
f(x) = \sum_{S\in {[n]\choose 3}}c_S \prod_{i\in S}x_i,
\eeq
where the~$c_S$ are some real coefficients.
For $i\in [n]$, define the $i$th \emph{slice} of~$f$ to be the symmetric matrix $M_i\in \R^{n\times n}$ with $(j,k)$-coefficient equal to~$c_{\{i,j,k\}}$ if $i,j,k$ are pairwise distinct and~$0$ otherwise. Then, define 
\beqn
\Delta(f) = \max_{i\in[n]} \|M_i\|.
\eeqn

The following is a slight variant of a decomposition due to Varopoulos~\cite{Varopoulos:1974}. 

\begin{lemma}[tri-linear Varopoulos decomposition]\label{lem:varopoulos}
Let~$f$ be an $n$-variate multilinear cubic form as in~\eqref{eq:3form}.
Then, for some $d\in \N$, there exist pairwise commuting matrices~$A_1,\dots,A_n\in \R^{d\times d}$ and orthogonal unit vectors~$u,v\in \R^d$ such that $\|A_i\|\leq 1$, and
\begin{align}
A_i^2 &=0\label{eq:var1}\\
\langle u, A_iv\rangle &=0\label{eq:var2}\\
\langle u, A_iA_jv\rangle &=0\label{eq:var3}\\
\langle u, A_iA_jA_kv\rangle &= \frac{c_{\{i,j,k\}}}{\Delta(f)}\label{eq:var4}
\end{align}
for all pairwise distinct $i,j,k\in [n]$.
\end{lemma}

\begin{proof}
For each $i\in[n]$, define~$M_i$ as above.
Define $W_i = \Delta(f)^{-1}M_i$ and note that this has operator norm at most~1.
For each $i\in[n]$, define the $(2n+2)\times(2n+2)$ block matrix
\beqn
A_i
=
{\footnotesize
\left[
\begin{array}{c|c|c|c}
\phantom{M}&&&\\
\hline
e_i&\phantom{M}&&\\
\hline
&W_i^\mathsf{T}&\phantom{M}&\\
\hline
&&e_i^\mathsf{T}&\phantom{M}
\end{array}
\right]},
\eeqn
where the first and last rows and columns have size~1, the second and third have size~$n$ and
where the empty blocks are filled with zeros.
Define $u = e_{2n+1}$ and $v = e_1$.
The rest of the proof is identical to the proof of~\cite[Lemma~2.11]{briet2018failure}, except for the property that~$A_i^2 = 0$.
This follows from the fact that
\beqn
A_i^2
=
{\footnotesize
\left[
\begin{array}{c|c|c|c}
\phantom{x}&&&\\
\hline
&\phantom{x}&&\\
\hline
W_i^\mathsf{T}e_i&&\phantom{x}&\\
\hline
&e_i^\mathsf{T}W_i^\mathsf{T}&&\phantom{x}
\end{array}
\right]}.
\eeqn
and since the $i$th row and $i$th column of~$M_i$ are zero.
\end{proof}

\begin{corollary}\label{cor:ccblb}
Let~$f$ be an $n$-variate multilinear cubic form as in~\eqref{eq:3form}.
Suppose that an $(n+2)$-variate quartic form $h\in \R[x_0,x_1,\dots,x_n,z]$ satisfies $h(x,1) = x_0f(x_1,\dots,x_n)$ for every $x\in \pmset{n+1}$.
Then,
\beqn
\|h\|_{\ccb} \geq \frac{\|f\|_2^2}{\Delta(f)}.
\eeqn
\end{corollary}

\begin{proof}
The multilinear monomials $\chi_S(x) = \prod_{i\in S}x_i$ with $S\subseteq \{0,\dots,n\}$ satisfy the orthogonality relations
\beq\label{eq:char_orth}
\Exp_{x\in \pmset{n+1}} \chi_S(x)\chi_T(x) = \delta_{S,T}.
\eeq
It follows that~$h$ and~$x_0f$ have equal coefficients for each quartic multilinear monomial in the variables $x_0,\dots,x_n$, which are~$c_S$ for~$x_0\chi_S$ with $S\in {[n]\choose 3}$ and~0 otherwise.
Let~$A_1,\dots,A_n\in \R^{d\times d}$ and $u,v\in \R^d$ be as in Lemma~\ref{lem:varopoulos} and let $A_0 = \id, A_{n+1} = 0$.
Commutativity and properties~\eqref{eq:var1}--\eqref{eq:var3} imply that if a quartic monomial expression $A_iA_jA_kA_l$ with $i,j,k,l\in \{0,\dots,n+1\}$ has repeated indices or an index equal to~$n+1$, then $\langle u, A_iA_jA_kA_lv\rangle = 0$.
With this, it  follows from property~\eqref{eq:var4} that
\begin{align}\label{eq:hfeq}
\langle u, h(A_0,\dots,A_{n+1})v\rangle
&=
\Big\langle u, \sum_{S\in {[n]\choose 3}} c_S\, A_0\chi_S(A_1,\dots,A_n)v\Big\rangle\\
&=
 \sum_{S\in {[n]\choose 3}} c_S \langle u, \chi_S(A_1,\dots,A_n)v\rangle\nonumber\\
&=
\Delta(f)^{-1}\sum_{S\in {[n]\choose 3}}c_S^2\nonumber\\
&=
\Delta(f)^{-1}\|f\|_2^2,\nonumber
\end{align}
where the last line is Parseval's identity~\cite[Chapter~1]{ODonnell:2009}.
\end{proof}

\section{Counterexamples}

Here, we prove Theorems~\ref{thm:counter4} and~\ref{thm:NoMult+AddConverse}.
But first we discuss the error in~\cite[pp.~920]{QQA=CBF}.
The proof there uses the equation
\begin{equation}\label{eq1}
    \sum_{\alpha,\beta\in\{0,1,2,3,4\}^{n}: |\alpha|+|\beta|=4}d'_{\alpha,\beta}x^\alpha=C\sum_{\alpha\in\{0,1\}^{n}: |\alpha|=4}d_\alpha x^\alpha\ \ \forall\ x\in\{-1,1\}^{n},
\end{equation}
where $d'_{\alpha,\beta}$, $d_\alpha$ and $C$ are real numbers and $|\alpha|$ stands for $\sum_{i=1}^n\alpha_i$. 
It follows from~\eqref{eq:char_orth} that $d'_{\alpha,0} = Cd_\alpha$ for all $\alpha\in\{0,1\}^n$ such that $|\alpha| = 4$.
What is used, however, is that $d'_{\alpha,0} = Cd_\alpha$ for all $\alpha\in\{0,1,2,3,4\}^n$ such that $|\alpha| = 4$, which is not true in general. 
For instance if $n=2$, $C=1$ and $d'_{(2,2),(0,0)}=d'_{(0,0),(4,0)}=-d'_{(2,0),(2,0)}=-d'_{(0,2),(2,0)}=1$ and the rest of the coefficients set to~$0$, then~\eqref{eq1} becomes $x_1^2x^2_2-x_1^2-x_2^2+1=0$.

Corollary~\ref{cor:ccblb} gets around this issue by using a multilinear cubic form instead of just a cubic form. This results in matrices~$A_i$ in Lemma~\ref{lem:varopoulos} that square to zero 
and has the effect that terms other than quartic multilinear monomials vanish in the left-hand side of~\eqref{eq:hfeq}.

\subsection{A random example}\label{subsec:randomEx}

The probabilistic proof of Theorem~\ref{thm:counter4} uses a random cubic form as in~\eqref{eq:3form} where the coefficients~$c_S$ are chosen to be independent uniformly distributed random signs.
Parseval's identity then gives $\|f\|_2^2 = {n\choose 3}$.
Each of the slices~$M_i$ of~$f$ is a random symmetric matrix with independent mean-zero entries of absolute value at most~1.
A standard random-matrix inequality and the union bound then imply that $\Delta(f) \leq C\sqrt{n}$ with probability $1 - \exp(-Cn)$~\cite[Corollary~2.3.6]{tao2012topics}.
By Hoeffding's inequality~\cite[Theorem~2.8]{boucheron:concentration} and the union bound, we have that $\|f\|_\infty \leq Cn^2$ with probability $1 - \exp(-Cn)$.
Rescaling~$f$ then gives that there exists a bounded multilinear cubic form such that $\|f\|_2^2/\Delta(f) \geq C\sqrt{n}$.
It now follows from Lemma~\ref{lem:qalgccb} with Corollary~\ref{cor:ccblb} that the $(n+1)$-variable quartic polynomial $x_0f(x_1,\dots,x_n)$ satisfies the requirements of Theorem~\ref{thm:counter4}.

\subsection{An explicit example}
We also give a constructive proof of Theorem~\ref{thm:counter4} using techniques from~\cite{briet2018failure}, which were used there to disprove a conjecture of Pisier on certain far-reaching generalizations of the Grothendieck inequality.
We do not exactly use the construction from that paper because it involves complex functions.
Instead, we will the M\"{o}bius function (defined below), which is real valued and has the desired properties. 

Let $n$ be a positive integer to be set later and
let $f_0:\Z_n\to [-1,1]$ be a function to be set later (where as usual~$\Z_n$ denotes the group of integers modulo~$n$).
Define $f$ to be the cubic multilinear form on~$3n$ variables given by
\beq\label{eq:APform}
f(x) = \sum_{a,b\in \Z_n}x(1,a)x(2, a+b)x(3,a+2b) f_0(a+3b),
\eeq
where we indexed the variables by~$[3]\times\Z_n$.

We claim that for some choice of~$f_0$, the quartic polynomial~$x_0f$, where~$x_0$ is an additional variable, meets the requirements of Theorem~\ref{thm:counter4}.
The generalized von Neumann inequality~\cite[Lemma~11.4]{TV06} allows us to bound the~$\infty$-norm of~$f$.
For a function $g:\Z_n\to \R$ and $b\in \Z_n$, define its multiplicative derivative $\Delta_bg:\Z_n\to\R$ to be the function $\Delta_bg(a) = g(a+b)g(a)$.
The Gowers 3-uniformity norm of~$g$ is then defined as
\beqn
\|g\|_{U^3} = \big(\Exp_{a,b,c,d\in \Z_n}\Delta_b\Delta_c\Delta_dg(a)\big)^{\frac{1}{8}}.
\eeqn

\begin{lemma}[generalized von Neumann inequality]\label{lem:gvn}
Suppose that~$n$ is coprime to~6. 
Then, for any function of the form~\eqref{eq:APform}, we have that
\beqn
\|f\|_\infty \leq n^2\|f_0\|_{U^3}.
\eeqn
\end{lemma}

The polynomial~$f$ has $3n$ slices, $M_{i,a}\in \R^{[3]\times\Z_n}$ for each $i\in[3]$ and $a\in \Z_n$, which we view as $3\times 3$ block-matrices with blocks indexed by~$\Z_n$.
The slice $M_{1,a}$ is supported only on the $(2,3)$ and $(3,2)$ blocks, which are each others' transposes.
On its $(2,3)$ block it has value $f_0(a+3b)$ on coordinate $(a+b, a+2b)$ for each~$b$.
In particular, this matrix has at most one nonzero entry in each row and column.
It follows that a relabeling of the rows turns $M_{1,a}$ into a diagonal matrix with diagonal entries in~$[-1,1]$, and therefore~$\|M_{1,a}\| \leq 1$.
Similarly, we get that $\|M_{i,a}\|\leq 1$ for $i=2,3$.
Hence, $\Delta(f) \leq 1$.
Parseval's identity implies that
\beqn
\|f\|_2^2 = n\sum_{a\in \Z_n}f_0(a)^2.
\eeqn

Identify~$\Z_n$ with $\{0,1\dots,n-1\}$ in the standard way.
We choose~$f_0$ to be the M\"{o}bius function restricted to this interval.
That is, set $f_0(0) = 0$ and for $a> 0$, set
\beqn
f_0(a) = 
\begin{cases}
1 & \text{if~$a$ is square-free with an even number of prime factors}\\
-1 & \text{if~$a$ is square-free with an odd number of prime factors}\\
0 & \text{otherwise.}
\end{cases}
\eeqn
Tao and Ter\"{a}v\"{a}inen~\cite{tao2021quantitative} recently proved that 
\beqn
\norm{f_0}_{U^3}\leq \frac{1}{(\log \log n)^C}
\eeqn
for some absolute constant~$C>0$.
It is also well-known that there are $\frac{6}{\pi^2}n-O(\sqrt{n})$  integers in~$[n]$ that are square-free \cite[page 269]{hardy1979introduction}. 
Normalizing~$f$ by $(\log\log n)^C/n^2$ and taking~$n$ coprime to~6 then gives a bounded multilinear cubic polynomial satisfying
\beqn
\frac{\|f\|_2^2}{\Delta(f)} \geq \frac{6}{\pi^2}(\log\log n)^C - o(1).
\eeqn
This proves Theorem~\ref{thm:counter4} as before.

\begin{remark}
The \emph{jointly completely bounded norm} of~$f$ is given by
\beqn
\|f\|_{\jcb} = \sup_{d\in \N} \|f(A_1,A_2, A_3)\|, 
\eeqn
where the supremum is taken over maps $A_1,A_2,A_3:\Z_n\to \C^{d\times d}$ such that  $\|A_i(a)\|\leq 1$ and $[A_i(a),A_j(b)] = [A_i(a),A_j(b)^*] = 0$ for all ${i\ne j}$ and $a,b\in \Z_n$.
Note that the only difference with the~$\ccb$ norm defined in Section~\ref{sec:preliminaries} is the second commutation relation involving the complex conjugates.
This norm can also be stated in terms of tensor products and the supremum is attained by observable-valued maps. 
As such, this norm appears naturally in the context of non-local games.
It was shown in~\cite{BBBLL:2019} that Proposition~\ref{lem:gvn} also holds for the jointly completely bounded norm, that is $\|f\|_{\jcb} \leq n^2\|f_0\|_{U^3}$.
The proof of Corollary~\ref{cor:ccblb} easily implies that $\|f\|_{\ccb} \geq \|f\|_2^2/\Delta(f)$.
This was used in~\cite{briet2018failure} to prove that the $\jcb$ and $\ccb$ norms are inequivalent.
\end{remark}

\subsection{SDPs for quantum query complexity}
Theorem~\ref{thm:NoMult+AddConverse} is based on an SDP for the completely bounded approximate degree of Gribling and Laurent~\cite{gribling2019semidefinite}.
The following notation will be convenient to state the SDP.
Let $\mathcal F(n,t)$ be the set of functions $f:[n]^t\to \R$ of the form
\beqn
f({\bf i}) = \langle u, A_1(i_1)\cdots A_t(i_t)v\rangle,
\eeqn
where $u,v\in S^{d-1}$ and $A_1,\dots,A_t:[n]\to \{M\in \R^{d\times d}\st \|M\|\leq 1\}$  for some~$d\in \N$.
A basic linear algebra argument shows that any such function can be obtained by setting $d = n^t$.
Given a function $\phi:\pmset{n}\to \R$, a sequence ${\bf i}\in [n+1]^t$ and setting $x_{n+1} = 1$, define
\beqn
\hat\phi({\bf i}) = \Exp_{x\in \pmset{n}}\phi(x) \prod_{j=1}^t x_{i_j}.
\eeqn
Note that if
\beqn
\phi(x) = \sum_{S\in {[n]\choose t}}c_S\chi_S(x)
\eeqn 
is a multilinear form of degree~$t$, then
\beq\label{eq:hatform}
\hat\phi(\bf i) = 
\begin{cases}
c_S & \text{if $\{i_1,\dots,i_t\} = S$}\\
0 & \text{otherwise.}
\end{cases}
\eeq
Given  $f:\{-1,1\}^n\to[-1,1]$ and $t\in\mathbb{N}$, define

	\begin{align} 
		\sdp(f,t) =  \max\ \ &\Exp_{x\in \{-1,1\}^n}\phi(x)f(x) - w\label{eq:SDP}\\
		\nonumber \mathrm{s.t.}\ \ &\phi: \{-1,1\}^n\to\mathbb{R},\ w\in \mathbb{R}\nonumber\\
		& \|\phi\|_1= 1\nonumber\\ 
		\nonumber
		\nonumber
		& (1/w)\hat\phi\in \mathcal F(n+1, t).\nonumber
	\end{align} 
Program \eqref{eq:SDP} corresponds to the optimization problem (24) of \cite{gribling2019semidefinite} for total functions and is written in a more convenient way for our purposes. 
There,~$f$  is considered to take values in $\{-1,1\},$ but their results still hold if $f$ is allowed to take values in $\mathbb{R}$, as we do here. Also, we must point out that the $A_i(i_s)$ used in the program (24) of \cite{gribling2019semidefinite} are unitaries, but there is no problem if we substitute them by contractions, thanks to the fact that every contraction can be seen as the top left corner of an unitary matrix \cite[Lemma 7]{Aaronson2015PolynomialsQQ}.
\jnote{Be more explicit about this?}

\begin{theorem}[Gribling-Laurent]\label{thm:Gribling-Laurent}
	If the optimal value of program \eqref{eq:SDP} is stricly larger than $\eps$, then there is no $\lfloor t/2\rfloor$-query algorithm $\mathcal{A}$ such that $|\mathbb{E}(\mathcal{A}(x))-f(x)|\leq\eps.$
\end{theorem}

\subsection{Approximation of quadratic forms}
Theorem \ref{query1} implies that bounded quadratic polynomials can be approximated by one-query quantum algorithms with error at most~$1-C$. 
Moreover, for $2n$-variate bounded bilinear forms $f(x,y) = x^\mathsf{T}Ay$ for $A\in \R^{n\times n}$, we can choose~$C$ to be $1/K_G{(n)}$, where $K_G{(n)}$ is the real Grothendieck constant of dimension~$n$ (see~\cite{QQA=CBF} for a short proof). Then, bounded bilinear forms can be approximated with an additive error of at most $1-1/K_G{(n)}$. Using Theorem \ref{thm:Gribling-Laurent}, we show that this is optimal for $n=2$, in which case $K_G(2)=\sqrt{2}$~\cite{fishburn1994bell}.

\begin{proposition}\label{prop:KG(2)=epsilon(2)}
	There exists a bilinear form~$f\in\mathbb{R}[x_1,x_2,x_3,x_4]$ such that there is no one-query quantum algorithm that approximates~$f$ on every $x\in\{-1,1\}^4$ with an additive error smaller than $1-1/\sqrt{2}$. 
\end{proposition}

\begin{proof}
We use the bilinear form that attains the Grothendieck constant of dimension 2, which is captured by the CHSH game. This form $f\in\mathbb{R}[x_1,x_2,x_3,x_4]$ is given by 
\beqn
f(x) = \frac{1}{2}\big(x_1(x_3+x_4) + x_2(x_3 - x_4)\big).
\eeqn
Clearly~$f$ maps $\{-1,1\}^4$ to $\{-1,1\}$,
and so $\norm{f}_1=\norm{f}_2^2=1$. 
We now emulate the construction from Lemma~\ref{lem:varopoulos}.
Writing the coefficients of~$f$ as $c_S$ for $S\in {[4]\choose 2}$, for each $i\in [4]$ define the unit vector $w_i\in \R^4$ by
\beqn
w_i = \sqrt{2}\sum_{j\in [4]\setminus \{i\}} c_{\{i,j\}}e_j.
\eeqn
Now define the matrices $A_i\in \R^{6\times 6}$ by
\beqn
A(i) = \begin{pmatrix}
	0   &   0   &   0\\
	w_i&   0   &    0\\
	0   &e_i^\mathsf{T} & 0
	\end{pmatrix}
\eeqn
It is easily verified that~$A(i)^2 = 0$ and that the $(6,1)$-coordinate of $A(i)A(j)$ equals~$\sqrt{2}c_{\{i,j\}}$ if $i\ne j$, from which it also follows that these matrices commute.
Setting $A(5) = 0$, we get that
\beqn
\langle e_6, A(i)A(j)e_1\rangle =
\begin{cases}
\sqrt{2}c_{\{i,j\}} & \text{if $\{i,j\}\in {[4]\choose 2}$}\\
0 & \text{otherwise.}
\end{cases}
\eeqn
Setting $\phi = f$ then gives that $\sqrt{2}\widehat{\phi} \in \mathcal F(5,2)$ and $\|\phi\|_1 = 1$.
This shows that $\sdp(f,2) \geq 1 - 1/\sqrt{2}$
\end{proof}

Proposition~\ref{prop:KG(2)=epsilon(2)} leads to a following natural question:

\begin{que}\label{que:KG=epsilon1}
	Is it true that for any $\eps>0$ there are an integer $n$ and a bounded bilinear form $f\in\mathbb{R}[x_1,\dots,x_{2n}]$ such that there is no one-query quantum algorithm that approximates $f$ on every $x\in\{-1,1\}^{2n}$ with an error smaller than $1-\frac{1}{K_G}-\epsilon$? 
\end{que}


\subsection{Approximation of cubic forms}
\label{sec:NoMultAddConverse}
Given that a generalization of Theorem \ref{query1} has been ruled out for quartic polynomials, one may wonder if a weaker converse for the polynomial method is possible: 

\begin{que}\label{que:mult+add}  Are there constants $C>0$ and $\eps>0$ such that for every bounded  polynomial $f$ of degree $4$ there is a $2$-query algorithm $\mathcal{A}$ such that $|\mathbb{E}(\mathcal{A}(x))-Cf(x)|< \eps$ for every $x\in\{-1,1\}^n$?
\end{que}

An affirmative answer to this question would imply that every polynomial of degree $4$ could be approximated by a $2$-query algorithm with additive error $1-C+\eps$. This would be the converse for the polynomial method that motivated Theorem \ref{query1} in \cite{Aaronson2015PolynomialsQQ}. 
Theorem~\ref{thm:NoMult+AddConverse} means that the $\varepsilon$ appearing in Question~\ref{que:mult+add} cannot be arbitrarily small. In other words, Theorem~\ref{thm:NoMult+AddConverse} says that there is no multiplicative converse even if we allow an (arbitrarily) small additive error. 

\begin{proof}[ of Theorem \ref{thm:NoMult+AddConverse}]
Let~$f\in \R[x_1,\dots,x_n]$ be a bounded multilinear cubic form as in~\eqref{eq:3form}.
As shown in the proof of Corollary~\ref{cor:ccblb}, there exist unit vectors $u,v\in \R^d$ and mappings $A:\{0,1,\dots,n+1\}\to \R^{d\times d}$ such that $\|A(i)\|\leq 1$ for each~$i$ and
\beqn
\big\langle u, A(i)A(j)A(k)A(l)v\big\rangle=
\begin{cases}
\frac{c_S}{\Delta(f)} & \text{if $\{i,j,k,l\} = \{0\}\cup S$ for $S\in {[n]\choose 3}$}\\
0 & \text{otherwise.}
\end{cases}
\eeqn
Let $g = x_0f/\norm{f}_\infty$.
Then, the function $\phi = x_0f/\|f\|_1$ meets the criteria of~\eqref{eq:SDP} with $w = \Delta(f)/\|f\|_1$ and shows that
\begin{align*}
	\sdp(g,4) &\geq 
	\frac{\|f\|_2^2}{\norm{f}_1\norm{f}_\infty} - \frac{\Delta(f)}{\|f\|_1}\\ 
	&\geq \frac{\|f\|_2^2}{\norm{f}_1\norm{f}_\infty}\Big(1 - \frac{\Delta(f)\norm{f}_\infty}{\norm{f}_2^2}\Big).
\end{align*}
If~$f$ is the random example from Section~\ref{subsec:randomEx}, then $\|f\|_2^2 = {n\choose 3}$ and $\Delta(f)\|f\|_\infty \leq Cn^{5/2}$ with high probability.
In particular, the above is positive for sufficiently large~$n$.
Similarly, for any $C\in (0,1)$ we get that $\sdp(Cg,4) > 0$ for sufficiently large~$n$.
The result now follows from Theorem~\ref{thm:Gribling-Laurent}.
\end{proof}

\section*{Acknowledgements}
We thank Yinan Li and anonymous referees for TQC'22 for helpful comments on a previous version of this paper.

\bibliographystyle{alphaabbrv}
\bibliography{Bibliography}
\end{document}